\documentclass[draftcls,onecolumn]{IEEEtran}

\ifCLASSINFOpdf
 \usepackage[pdftex]{graphicx}
 \DeclareGraphicsExtensions{.pdf,.jpeg,.png,.eps}
\else
\fi

\usepackage{amsmath,amsfonts,amssymb}
\usepackage{sublabel}
\usepackage{amsthm}
\usepackage{wasysym}
\usepackage{subfigure}
\usepackage{hyperref}

\newtheorem{theorem}{Theorem}
\newtheorem{definition}{Definition}

\def\v#1{\underline{#1}}  

\def\expect#1#2{\mathrm{E}_{#1}\left[#2\right]}
\newcommand{\NT}{N_{\rm T}}
\newcommand{\NR}{N_{\rm R}}

\newcommand{\her}{^{\rm H}}

\newcommand{\Id}{{\bf I}}

\begin{document}
%
\title{Receive Diversity and Ergodic Performance of Interference Alignment on the MIMO Gaussian Interference Channel}

\author{\IEEEauthorblockN{Maxime Guillaud\\}
\IEEEauthorblockA{Institute of Communications and Radio-Frequency Engineering\\
Vienna University of Technology\\
Vienna, Austria\\
Email: {\small \href{mailto:guillaud@tuwien.ac.at}{\texttt{guillaud@tuwien.ac.at}}}}}



%


\maketitle

\begin{abstract}
We consider interference alignment (IA) over $K$-user Gaussian MIMO interference channel (MIMO-IC) when the SNR is not asymptotically high. We introduce a generalization of IA which enables receive diversity inside the interference-free subspace. We generalize the existence criterion of an IA solution proposed by Yetis et al. to this case, thereby establishing a multi-user diversity-multiplexing trade-off (DMT) for the interference channel.
Furthermore, we derive a closed-form tight lower-bound for the ergodic mutual information achievable using IA over a Gaussian MIMO-IC with Gaussian i.i.d. channel coefficients at arbitrary SNR, when the transmitted signals are white inside the subspace defined by IA. Finally, as an application of the previous results, we compare the performance achievable by IA at various operating points allowed by the DMT, to a recently introduced distributed method based on game theory.
\end{abstract}


%
\IEEEpeerreviewmaketitle

\section{Introduction}
\label{section_introduction}

Interference alignment (IA) was first considered in \cite{Jafarit08} as a coding technique for the two-user Multiple-Input Multiple-Output (MIMO) X channel.  The degrees of freedom (DoF) region for this channel has been analyzed in \cite{Jafarit08} for an arbitrary number of antennas per user $M>1$, and IA was shown to achieve the maximum $\frac{4}{3}M$ degrees of freedom achievable on this channel, based only on linear precoding at the transmitters and zero-forcing at the receivers. With IA, thanks to the alignment of all interfering signals in the same subspace from the point of view of each receiver, interference can be removed simply through zero-forcing filtering.

The scheme was later generalized to the $K$-user interference channel \cite{Cadambeit08}, where it was shown to achieve almost surely a sum-rate multiplexing gain of $\frac{K}{2}$ per time, frequency and antenna dimension. In comparison, independent operation of $K$  \emph{isolated} point-to-point links would incur a sum-rate multiplexing gain of $K$ per dimension. This indicates that IA allows virtually interference-free communications, at the cost of halving the multiplexing gains with respect to what the users could achieve over isolated point-to-point links.

In the $K$-user Gaussian MIMO IC, under mild hypotheses on the distribution of the channel coefficients, the existence with probability 1 of a solution to the IA problem depends only on the dimensions of the problem (number of users $K$ and number of antennas at each node). An existence criterion was introduced in  \cite{Yetis_Jafar_feasibility_conditions_IA_Globecom09}. An iterative algorithm was introduced in \cite{Gomadamit08} to find numerically the precoding matrices achieving IA. Closed-form solutions are available for certain particular cases (e.g. when all nodes have $N=K-1$ antennas, \cite{Tresch_Guillaud_Riegler_SSP09}).
 
Although the high-SNR properties of IA (in the form of the DoF) have been thoroughly studied, little is known about how fast the performance of IA degrades in the presence of noise.
There have been attempts reported in the literature to apply IA to systems where the SNR is not asymptotically high, despite the fact that IA is arguably suboptimal in that case (this is mostly due to the relative ease of implementation of IA w.r.t. other linear beamforming techniques such as \cite{Shi_Honig_Utschick_interference_pricing_ICC09}, where the optimization of e.g. sum-rate is performed). Evaluation of the mutual information (or of the sum-rate) achieved by IA by means of Monte-Carlo simulations \cite{Tresch_Guillaud_PIMRC09} or using measured channels \cite{ElAyach_etal_IA_over_measured_channels_VT10} have been performed. However, the theoretical performance limits of IA in this case remain largely unknown.\\

The objective of this paper is to explore the application of IA to the MIMO-IC at non-asymptotic SNR. Our contributions are as follows:
\begin{enumerate}
\item We generalize IA to the case where the receiver exploits spatial diversity inside the interference-free subspace created by IA. Specifically, we study the situation where the codimension of the interference subspace at the receiver is not equal to the number of transmit DoF, and generalize the existence criterion of Yetis et al. \cite{Yetis_Jafar_feasibility_conditions_IA_Globecom09} to that case. This enables the definition of a diversity-multiplexing trade-off (DMT) for the interference channel.
\item We derive a closed-form expression for the ergodic mutual information achievable using IA over a Gaussian MIMO-IC with Gaussian i.i.d. channel coefficients. The formula is valid for arbitrary sets of transmit powers and channel noise variances, i.e. are not restricted to any SNR range, for the case where each node transmits a spatially white signal (inside the chosen subspace). It enables the asymptotic analysis (in the number of users and antennas) of interference-aligned systems, which was not possible previously due to the intractable complexity of Monte-Carlo methods applied to large systems, and enables comparisons between the various operating points allowed by the DMT.
\item We compare numerically the performance of IA to distributed algorithms, providing insight on their relative performance.\\
\end{enumerate}

This article is organized as follows: the system model and the definition of interference alignment for the Gaussian MIMO-IC are introduced in Section \ref{section_sysmodel}. Section \ref{section_non_square_eqchannel} generalizes the definition of IA to include receive diversity, and introduces a generalized achievability criterion.
In Section \ref{section_ergodic_MI_for_IA}, we derive closed-form formulas for the ergodic mutual information achieved by the scheme developed previously, for two types of receiver-side processing, and validate them through simulations.
Section \ref{section_applications} presents a comparison of the performance (in terms of ergodic rates) of IA as defined before, for various DMT points, to that achieved with the distributed transmit covariance optimization approach recently proposed in \cite{Scutari_MIMO_Iterative_WTF_IT09}.

Notation: In the sequel, $\expect{\mathrm{X}}{\cdot}$ denotes the expectation operator over $\mathrm{X}$, while $(\cdot)\her$ denotes the Hermitian transpose operation.

\section{System Model}
\label{section_sysmodel}
In this section, we first present our system model and recall the definition of interference alignment over the MIMO-IC.

We consider a $K$-user MIMO interference channel where transmitter $j=1\ldots K$ is equipped with $M_j$ antennas and receiver $i$ with $N_i$ antennas. The MIMO channel is assumed to be frequency-flat. Here, we will consider the case where the channel between transmitter $j$ and receiver $i$, denoted by the $N_i\times M_j$ matrix $\mathbf{H}_{ij}$, is Rayleigh fading, i.e. the elements of $\mathbf{H}_{ij}$ are complex Gaussian i.i.d. random variables, with zero mean and unit variance.

Let us focus on the $i^{\textrm{th}}$ receiver $(1\leq i \leq K)$, which receives interference from other transmitters $j\neq i$ in addition to its intended signal.
The discrete-time system model is given by
\begin{equation}\label{equ_yibar}
\mathbf{y}_i =  \mathbf{H}_{ii}\mathbf{V}_{i} \mathbf{s}_i +
\sum_{j=1,j\neq i}^{K}\mathbf{H}_{ij} \mathbf{V}_{j} \mathbf{s}_j + \mathbf{n}_i,
\end{equation}
where $\mathbf{s}_i \in \mathbb{C}^{d_i\times 1}$ is a vector representing the signal from transmitter $i$, and $\mathbf{V}_{i}\in \mathbb{C}^{M_i\times d_i}$ is the precoding matrix at transmitter $i$. Furthermore, $\mathbf{n}_i$ accounts for the thermal noise generated in the radio frequency front-end of the receiver and interference from sources other than the interfering transmitters.

IA is achieved with degrees of freedom $(d_1,\ldots, d_K)$ (where each $d_i$ corresponds to the multiplexing gain achieved for a transmitter-receiver pair, i.e. $d_i$ streams per transmitter are spatially pre-coded at transmitter $i$) iff there exist $M_i\times d_i$ truncated unitary matrices\footnote{Note that the requirement that the columns of $\mathbf{U}_{i}$ and $\mathbf{V}_{i}$ be orthonormal was not present in the original definition of IA in \cite{Gomadamit08}. However it is easy to see that, since they must be full column rank as per eq.~\eqref{equ_IAcondB}, the $\mathbf{R}$ factors in the QR decompositions of all $\mathbf{U}_{i}$'s and $\mathbf{V}_{i}$'s must be invertible, and can therefore be discarded without changing the alignment criterion in \eqref{equ_IAcondA}. In the sequel, we assume w.l.o.g. that all $\mathbf{U}_{i}$'s and $\mathbf{V}_{i}$'s are truncated unitary matrices.} (precoding matrices) $\mathbf{V}_{i}$ and $N_i\times d_i$ truncated unitary matrices (zero-forcing interference suppression matrices) $\mathbf{U}_{i}$ such that, for $i=1,\hdots,K$,
\begin{align}
\mathbf{U}_{i}\her \mathbf{H}_{ij}\mathbf{V}_{j} &= 0, \quad  \forall j\neq i, \label{equ_IAcondA}\\
\textrm{rank}\left(\mathbf{U}_{i}\her \mathbf{H}_{ii}\mathbf{V}_{i} \right)&= d_i, \quad  \forall i. \label{equ_IAcondB}
\end{align}

An iterative algorithm \cite[Algorithm 1]{Gomadamit08}, which is based on the minimization of an interference leakage metric (zero leakage is equivalent to the system of equations in \eqref{equ_IAcondA} and \eqref{equ_IAcondB}), was introduced to find the precoding matrices and to verify the achievability of interference channel settings. At every iteration, the algorithm from \cite{Gomadamit08} involves the computation of $K$ eigenvalue problems. Furthermore, depending on the interference channel setting the convergence speed can vary significantly. Nevertheless, the iterative algorithm provides numerical insight into the feasibility of IA for the $K$-user  MIMO interference channel with arbitrary channel dimensions and for any IA multiplexing gains.

If a solution to the system of equations \eqref{equ_IAcondA}-\eqref{equ_IAcondB} is found, we define the projection receiver as a receiver where $\mathbf{y}_i$ is projected onto the interference-free subspace before decoding. In other words, the decoder operates on $\bar{\mathbf{y}}_i = \mathbf{U}_{i}\her \mathbf{y}_i$. Note that 
\begin{align}\label{equ_yibar_zf}
\bar{\mathbf{y}}_i &= \mathbf{U}_{i}\her\mathbf{H}_{ii}\mathbf{V}_{i} \mathbf{s}_i + \sum_{j\neq i}\mathbf{U}_{i}\her\mathbf{H}_{ij} \mathbf{V}_{j} \mathbf{s}_j + \mathbf{U}_{i}\her\mathbf{n}_i,\nonumber \\
              &= \mathbf{U}_{i}\her\mathbf{H}_{ii}\mathbf{V}_{i} \mathbf{s}_i+ \mathbf{U}_{i}\her\mathbf{n}_i,
\end{align}
i.e. the interference terms present in \eqref{equ_yibar} are perfectly suppressed, due to \eqref{equ_IAcondA}. Note that the energy of the signal part that lies in the interference subspace is lost. However, this power loss is irrelevant in the context of the degrees of freedom analysis. \\

\section{Multi-User Diversity-Multiplexing Trade-Off over the Interference Channel}
\label{section_non_square_eqchannel}

In this section, we extend the arguments and notations of \cite{Yetis_Jafar_feasibility_conditions_IA_Globecom09} regarding the existence of a solution to the IA problem to the case where the codimension of the interference subspace at the receiver is not equal to the number of transmit DoF, permitting for extra receive diversity.

In order to do this, we let $\mathbf{U}_{i}$ and $\mathbf{V}_{i}$ have non-equal numbers of columns. Let $d_i$ denote the number of columns of the $i$-th transmitter-side precoder $\mathbf{V}_{i}$, as before, while $d'_i$ denotes the number of columns of $\mathbf{U}_{i}$. For any $M_i \geq d_i\geq 1$ and $N_i \geq d'_i\geq 1$,  $\textrm{rank}\left(\mathbf{U}_{i}\her \mathbf{H}_{ii}\mathbf{V}_{i} \right) \leq \min(d_i, d'_i)$, since the matrix $\mathbf{U}_{i}\her \mathbf{H}_{ii}\mathbf{V}_{i}$ is of dimensions $d'_i \times d_i$.
Therefore, from a DoF point of view, it is clear that taking $d'_i=d_i$ is optimal since $d'_i<d_i$ would not permit to achieve the desired $d_i$ degrees of freedom for user $i$, while $d'_i>d_i$ would impose more constraints (reduce the set of problem dimensions for which a solution exists) without improving the achieved DoF.
However, if one considers non-asymptotic SNR situations, the DoF metric is not the only relevant metric anymore. In particular, taking $d'_i>d_i$ ensures that the interference-free equivalent channel $\mathbf{U}_{i}\her\mathbf{H}_{ii}\mathbf{V}_{i}$ in \eqref{equ_yibar_zf} is tall, i.e. permits to increase the available receive diversity.

The IA conditions represented by the matrix equality \eqref{equ_IAcondA} can be rewritten as a set of scalar equality equations: for $1\leq i\neq j \leq K$,
\begin{equation} 
  {\v{u}_m^{[i]}}\her \mathbf{H}_{ij}\v{v}_n^{[j]} =0, \forall (n,m) \in \{1,\ldots,d_j\}\times \{1,\ldots,d'_i\}, \label{Eijmn}
\end{equation}
where $\v{u}_m^{[i]}$ and $\v{v}_n^{[j]}$ are respectively the $m$th column of $\mathbf{U}_{i}$ and the $n$th column of $\mathbf{V}_{j}$. Equation \eqref{Eijmn} is denoted $E_{ij}^{mn}$. 

The rest of the proof proceeds in a fashion similar to \cite{Yetis_Jafar_feasibility_conditions_IA_Globecom09}: denoting by $\mathrm{var}(E_{ij}^{mn})$ the set of non-redundant variables involved in \eqref{Eijmn}, its cardinality is
\begin{equation} \label{eq_cardinality_Eijmn}
\left|\mathrm{var}(E_{ij}^{mn})\right| = (M_j-d_j)+(N_i-d'_i).
\end{equation}
The complete set of scalar equations equivalent to \eqref{equ_IAcondA} is therefore 
{\small
\begin{equation}
\mathcal{E} = \left\{ E_{ij}^{mn} | 1\leq i\neq j \leq K, n \in \{1,\ldots,d_j\},  m \in \{1,\ldots,d'_i\}\right\}.
\end{equation} }

We now recall the definition of a proper system from \cite{Yetis_Jafar_feasibility_conditions_IA_Globecom09}:
\begin{definition}
\label{def_proper_nonsymmetric}
The system of equation $\mathcal{E}$  is proper iff for any $S \subset \mathcal{E}$,
\begin{equation} \label{eq_definition_proper}
 |S| \leq \left| \bigcup_{E\in S}  \mathrm{var}(E)\right|.
\end{equation}
In other words, a system is proper iff for all subsets of equations, the number of variables involved is at least as large as the number of equations.
\end{definition}

\cite[Section IV]{Yetis_Jafar_feasibility_conditions_IA_Globecom09} argues that proper systems  admit a solution almost surely, albeit without providing a formal proof of this claim. Nevertheless, this criterion has been found experimentally to be reliable, and we will therefore rely on it as well in the sequel.

We now particularize the existence criterion for our case of interest, namely symmetric systems with $d'_i\neq d_i$.

\begin{theorem}[Symmetric system]
\label{thm_feasibility_D_Dp_symmetric}
A symmetric systems where all transmitters have the same number of antennas $M_j=\NT, \, \forall j$ and all receivers have the same number of antennas $N_i=\NR, \, \forall i$, and where $\mathbf{V}_{i}$ and $\mathbf{U}_{i}$ have respectively $d_i=d$ and $d'_i=d'$ columns for all users,  is proper iff 
\begin{equation}
d(\NT-d)+d'(\NR-d')-d d' (K-1)\geq 0. \label{sym_proper_condition}
\end{equation}
\end{theorem}
\begin{proof}
We follow again the method outlined in \cite{Yetis_Jafar_feasibility_conditions_IA_Globecom09}. A symmetric is proper iff the total number of variables $N_v=|\mathrm{var}(\mathcal{E})|$ is equal to or greater than the total number of equations $N_e=|\mathcal{E}|$. Since
\begin{eqnarray} \label{eq_proper_symmetriccase}
 N_v &=& \sum_{k=1}^K d_k\left( M_k-d_k\right) + d'_k\left( N_k-d'_k\right) \\
 &=& K\left[ d(\NT-d)+d'(\NR-d')\right],
\end{eqnarray}
and
\begin{eqnarray}
 N_e &=& \sum_{\scriptsize \begin{array}{c} 1\leq j,k\leq K\\ j\neq k \end{array}} d'_k d_j \\
 &=& K (K-1) d d',
\end{eqnarray}
$N_v \geq N_e$ is equivalent to \eqref{sym_proper_condition}.
\end{proof}

Essentially, Definition~\ref{def_proper_nonsymmetric} (or Theorem~\ref{thm_feasibility_D_Dp_symmetric} for symmetric systems), together with our new definition of IA including receive diversity in \eqref{Eijmn}, establish a multi-user DMT applicable to the interference channel, in a way reminiscent of the result of Zheng and Tse \cite{zheng_tse_diversity_multiplexing}. If the system is proper (and thus IA is feasible), user $i$ enjoys multiplexing gain $d_i$ and receive diversity $d'_i-d_i+1$. For a given number of users and channel sizes, condition \eqref{eq_definition_proper} (or \eqref{eq_proper_symmetriccase} for the symmetric case) governs how the channel degrees of freedom can be traded-off between diversity, multiplexing, and between users.\\

\section{Ergodic Mutual Information for Interference Alignment}
\label{section_ergodic_MI_for_IA}
In this section, we introduce closed-form formulas for the ergodic mutual information achieved by IA over the Gaussian MIMO-IC, when the channel coefficients are Gaussian i.i.d. Here, we assume that the additive noise $\mathbf{n}_i$ is complex Gaussian circularly symmetric, with covariance matrix $\sigma^2\Id$.

\subsection{Optimum Receiver}
\label{section_ergodic_MI_optimumreceiver}
Conversely to previous results \cite{Guillaud_Tresch_NewcomWS_BCN_09} based on the equivalent channel (where the part of the signal from the transmitter of interest lying in the interference subspace is discarded, see \eqref{equ_yibar}), we characterize here the mutual information $I(\mathbf{s}_k;\mathbf{y}_k|\mathbf{H})$. In cases where the number or the power of the interferers is low, this mutual information can be significantly higher than $I(\mathbf{s}_i;\bar{\mathbf{y}}_i|\mathbf{H})$.

We build upon a recent result of Chiani et al. \cite{Chiani_Ergodic_MI_intf_networks_IT10}, where the ergodic mutual information of a single-user MIMO link is characterized analytically for Rayleigh fading. Let us recall it here:

\begin{theorem}{(\cite{Chiani_Ergodic_MI_intf_networks_IT10}, Theorem 1)}
Let $\mathcal{C}_{\mathrm{SU}} \left( n,p,\Phi \right) = \expect{\mathrm{\mathbf{H}}}{ \log \det \left( \mathrm{\mathbf{I}}_{p} + \mathrm{\mathbf{H}} \Phi \mathrm{\mathbf{H}}^H\right) }$,
where $\mathrm{\mathbf{H}}$ is a $p\times n$ matrix with complex Gaussian i.i.d. coefficients of unit variance, and $\Phi$ is any positive definite matrix of dimension $n \times n$.
Letting $\mu_{1}> \mu_{2} > \ldots > \mu_{L}$ denote the $L$ distinct eigenvalues of $\Phi^{-1}$, and $m_1,\ldots, m_L$ their respective multiplicities ($\sum_{i=1}^L m_i=n$),
\begin{equation}
 \mathcal{C}_{\mathrm{SU}} \left( n,p,\Phi \right) = K \sum_{k=1}^{n_\mathrm{min}} \det\left( \mathbf{R}^{(k)}\right),
\end{equation}
where $n_\mathrm{min}=\min(n,p)$,
\begin{equation}
 K = \frac{(-1)^{p(n-n_\mathrm{min})}}{\Gamma_{n_\mathrm{min}}(p)}\frac{\prod_{i=1}^L \mu_{i}^{m_i p}}{\prod_{i=1}^L \Gamma_{(m_i)}(m_i) \prod_{i<j}(\mu_{i}-\mu_{j})^{m_i m_j}},
\end{equation}
$\mathbf{R}^{(k)}$ is a $n\times n$ matrix with $i,j$-th element
\begin{equation}
  r_{i,j}^{(k)} = \left\{ \begin{array}{l}
  (-1)^{a_i}\int_0^\infty x^{p-n_\mathrm{min}+j-1+a_i} e^{-x \mu_{e_i}} \mathrm{d} x \\
  \hspace{3cm} \mathrm{for} \, j=1,\ldots,n_\mathrm{min}, j\neq k,\\
  (-1)^{a_i}\int_0^\infty x^{p-n_\mathrm{min}+j-1+a_i} e^{-x \mu_{e_i}} \log(1+x) \mathrm{d} x \\
  \hspace{3cm} \mathrm{for} \, j= k,\\
  (n-j)!/(n-j-a_i)! \mu_{e_i}^{n-j-a_i} \\
  \hspace{3cm} \mathrm{for} \, n_\mathrm{min}+1\leq j \leq n,
  \end{array} \right.
\end{equation}
$e_i$ is the unique integer such that $m_1+\ldots+m_{e_i-1} < i \leq m_1+\ldots+m_{e_i}$, and $a_i = \sum_{k=1}^{e_i} m_k -i$.
\end{theorem}

We now derive an approximate closed-form lower bound (shown to be tight in Section \ref{section_simuls_validation}) for the ergodic mutual information achieved by interference alignment in the case where the optimum (single-user) decoder is used at the receiver.

\begin{theorem} \label{thm_bound_IA}
The ergodic mutual information achieved by the $k$-th user of a Gaussian MIMO IC with IA where user $i$ splits its transmit power $P_i$ evenly among $d_i$ uncorrelated transmit dimensions, is approximately lower-bounded as
\begin{equation}
\expect{\mathbf{H}}{I(\mathbf{s}_k;\mathbf{y}_k|\mathbf{H})} \apprge  \mathcal{C}_{\mathrm{SU}} \left( N_k,d_k,\Psi_k \right), \label{eq_thm_MI_bound}
\end{equation}
where $\Psi_k=\left[\begin{array}{cc} \frac{P_k}{d_k\sigma^2} \Id_{d'_k} & \mathbf{0}\\ \mathbf{0} & \frac{P_k}{d_k(\sigma^2 +\sum_{i\neq k} P_i)}  \Id_{N_k-d'_k} \end{array} \right]$.
\end{theorem}

\begin{proof}

Let us assume that the transmission scheme of \eqref{equ_yibar} is in use, with unitary, orthogonal precoding vectors fulfilling the conditions of eq.~\eqref{equ_IAcondA}.
Focusing on the signal received by user $k$, let $\mathbf{H}_0=\mathbf{H}_{k,k}$ denote the direct channel of the user of interest, $\mathbf{Q}_0=\mathbf{Q}_k$ the covariance of its transmitted signal, and let $\mathbf{H}_I=\left[\mathbf{H}_{k,1}\ldots \mathbf{H}_{k,k-1} ,\mathbf{H}_{k,k+1}\ldots \mathbf{H}_{k,K}\right]$ be the matrix containing the channel coefficients of the $K-1$ interfering links (note that we drop the dependency on receiver index $k$ for notational simplicity). Furthermore, by assumption, all users are transmitting spatially white signals of respective powers $P_i$, i.e. $\mathbf{Q}_k=\frac{P_k}{d_k}\Id_{d_k}$. 

The mutual information for user $k$ is
{\small 
\begin{eqnarray}
I(\mathbf{s}_k;\mathbf{y}_k|\mathbf{H})&=& \log\frac{\det\left(\sigma^2\Id_{N_k} +\mathbf{H}_0 \mathbf{V}_0\mathbf{Q}_0\mathbf{V}_0\her\mathbf{H}_0\her + \sum_{i\neq k} \mathbf{H}_{k,i}\mathbf{V}_{i}\mathbf{Q}_i\mathbf{V}_{i}\her\mathbf{H}_{k,i}\her\right)}{\det\left(\sigma^2\Id_{N_k} + \sum_{i\neq k} \mathbf{H}_{k,i}\mathbf{V}_{i}\mathbf{Q}_i\mathbf{V}_{i}\her\mathbf{H}_{k,i}\her\right)}  \nonumber\\
&=&\log\det\left(\Id_{N_k} + \mathbf{H}_0\mathbf{V}_0\mathbf{Q}_0\mathbf{V}_0\her\mathbf{H}_0\her (\sigma^2\Id_{N_k} +\mathbf{H}_I\mathbf{V}_I\mathbf{Q}_I\mathbf{V}_I\her\mathbf{H}_I\her)^{-1}\right).
\end{eqnarray} }
where $\mathbf{V}_{I}=\mathrm{diag}(\mathbf{V}_1, \ldots, \mathbf{V}_{k-1}, \mathbf{V}_{k+1},\ldots, \mathbf{V}_{K})$, and $\mathbf{Q}_I=\mathrm{diag}(\mathbf{Q}_1, \ldots, \mathbf{Q}_{k-1}, \mathbf{Q}_{k+1},\ldots, \mathbf{Q}_{K})$.

According to the interference alignment conditions \eqref{equ_IAcondA}, the $d'_k$ orthonormal columns of $\mathbf{U}_k$ span a $d'_k$-dimensional subspace which remains free of interference. The sum of all the interference terms at receiver $k$ is therefore contained within a subspace of dimension $N_k-d'_k$. Let $\mathbf{U}_k^\bot$ denote a $N_k\times N_k-d'_k$ matrix containing an orthonormal basis of this interference subspace.

Furthermore, let $\mathbf{U}=\left[ \mathbf{U}_k \mathbf{U}_k^\bot \right]$. Since $\mathbf{U}$ is unitary, we can rewrite
\begin{eqnarray}
I(\mathbf{s}_k;\mathbf{y}_k|\mathbf{H}) & = & \log\det\left( \mathbf{U}\her \left( \Id_{N_k} + \mathbf{H}_0\mathbf{V}_0\mathbf{Q}_0\mathbf{V}_0\her\mathbf{H}_0\her \mathbf{U}\mathbf{U}\her (\sigma^2\Id_{N_k} +\mathbf{H}_I\mathbf{V}_I\mathbf{Q}_I\mathbf{V}_I\her \mathbf{H}_I\her)^{-1}\right) \mathbf{U} \right) \\
&=& \log\det\left( \Id_{N_k} + \mathbf{U}\her\mathbf{H}_0\mathbf{V}_0\mathbf{Q}_0\mathbf{V}_0\her\mathbf{H}_0\her \mathbf{U}  \mathbf{U}\her(\sigma^2\Id_{N_k} +\mathbf{H}_I\mathbf{V}_I\mathbf{Q}_I\mathbf{V}_I\her\mathbf{H}_I\her)^{-1} \mathbf{U} \right) \\
&=& \log\det\left( \Id_{N_k} + \mathbf{U}\her\mathbf{H}_0\mathbf{V}_0\mathbf{Q}_0\mathbf{V}_0\her\mathbf{H}_0\her \mathbf{U} (\sigma^2\Id_{N_k} + \mathbf{U}\her\mathbf{H}_I\mathbf{V}_I\mathbf{Q}_I\mathbf{V}_I\her\mathbf{H}_I\her\mathbf{U})^{-1}  \right),
\end{eqnarray}
where we used the fact that $\mathbf{U}^{-1}=\mathbf{U}\her$.
Since we are interested in the ergodic mutual information $\expect{\mathbf{H}}{I(\mathbf{s}_k;\mathbf{y}_k|\mathbf{H})}$, we need to consider the statistics of the random variables involved in the above equation. Note in particular that $\mathbf{U}$ and $\mathbf{V}_0$ are implicit functions (through \eqref{equ_IAcondA}) of $\mathbf{H}_{i,j}$, $1 \leq i\neq j \leq K$. However, they are independent of $\mathbf{H}_0$. Since $\mathbf{H}_0$ is Gaussian i.i.d. distributed, and since this distribution is invariant by multiplication with an independent unitary matrix, $\mathbf{U}\her\mathbf{H}_0\mathbf{V}_0$ is a $N_k \times d_k$ matrix with complex Gaussian i.i.d. coefficients of unit variance \cite{Muirhead_multivariate_theory}. Letting $\bar{\mathbf{H}}_0$ denote a matrix with the same distribution, we can rewrite the expectation over $\mathbf{H}$ as
\begin{eqnarray}
\expect{\mathbf{H}}{I(\mathbf{s}_k;\mathbf{y}_k|\mathbf{H})} & = & \mathrm{E}_{\bar{\mathbf{H}}_0,\mathbf{H}_I} \left[\log\det\left( \Id_{N_k} + \bar{\mathbf{H}}_0\mathbf{Q}_0\bar{\mathbf{H}}_0\her (\sigma^2\Id_{N_k} +\mathbf{U}\her\mathbf{H}_I\mathbf{V}_I\mathbf{Q}_I\mathbf{V}_I\her\mathbf{H}_I\her\mathbf{U})^{-1}  \right) \right].
\end{eqnarray}

Let us now focus on the rotated covariance of the interference plus noise term, $\mathbf{K}_I = \sigma^2\Id_{N_k} +\mathbf{U}\her\mathbf{H}_I\mathbf{V}_I\mathbf{Q}_I\mathbf{V}_I\her\mathbf{H}_I\her\mathbf{U}$.
In particular, note that $\mathbf{U}\her\mathbf{H}_I\mathbf{V}_I=\left[\begin{array}{c} \mathbf{U}_k\her \mathbf{H}_I\mathbf{V}_I \\ {\mathbf{U}_k^\bot}\her \mathbf{H}_I\mathbf{V}_I\end{array} \right] = \left[\begin{array}{c} \mathbf{0} \\ {\mathbf{U}_k^\bot}\her \mathbf{H}_I\mathbf{V}_I\end{array} \right]$, since  $\mathbf{U}_k\her \mathbf{H}_I\mathbf{V}_I$ is a block matrix where each block $\mathbf{U}_k\her \mathbf{H}_{k,i} \mathbf{V}_i$ is zero according to \eqref{equ_IAcondA}. The statistics of ${\mathbf{U}_k^\bot}\her \mathbf{H}_I\mathbf{V}_I$ are hard to come by, since the terms of the matrix products are implicitly dependent through \eqref{equ_IAcondA}.
As an approximation, we will make the assumption that ${\mathbf{U}_k^\bot}\her \mathbf{H}_I\mathbf{V}_I$ behaves like a Gaussian i.i.d. distributed random matrix $\mathbf{H}_I'$ with unit variance per component (this would be exact if $\mathbf{U}_k^\bot$ and $\mathbf{V}_i$ were independent of $\mathbf{H}$).
In that case, $\expect{\mathbf{H}_I}{\mathbf{K}_I} \approx  \mathbf{Q}_Z$, where
\begin{eqnarray}
\mathbf{Q}_Z & = & \expect{\mathbf{H}_I'}{\sigma^2\Id_{N_k} +\left[\begin{array}{c} \mathbf{0} \\ \mathbf{H}_I'\end{array} \right]\mathbf{Q}_I \left[\begin{array}{c} \mathbf{0} \\ \mathbf{H}_I'\end{array} \right]\her} \\
 &=&\left[\begin{array}{cc} \sigma^2 \Id_{d'_k} & \mathbf{0}\\ \mathbf{0} & (\sigma^2 +\sum_{i\neq k} P_i) \Id_{N_k-d'_k} \end{array} \right]. 
\end{eqnarray} 

Using those notations, one obtains that 
\begin{eqnarray}
\expect{\mathbf{H}_I}{I(\mathbf{s}_k;\mathbf{y}_k|\mathbf{H})}& = &  \expect{\mathbf{H}_I}{\log\det\left( \Id_{N_k} + \bar{\mathbf{H}}_0\mathbf{Q}_0\bar{\mathbf{H}}_0\her\mathbf{K}_I^{-1}\right)} \\
& \geq & \log\det\left( \Id_{N_k} + \bar{\mathbf{H}}_0\mathbf{Q}_0\bar{\mathbf{H}}_0\her\left(\expect{\mathbf{H}_I}{\mathbf{K}_I}\right)^{-1}\right) \\
& \gtrapprox & \log\det\left( \Id_{N_k} + \bar{\mathbf{H}}_0\mathbf{Q}_0\bar{\mathbf{H}}_0\her \mathbf{Q}_Z^{-1}\right).  \label{approxgeq}
\end{eqnarray}
where the lower bound follows from the convexity in $\mathbf{X}$ of $\log\det\left(\Id+\mathbf{K}\mathbf{X}^{-1} \right)$ for positive definite matrices $\mathbf{K}$ and $\mathbf{X}$ \cite[Lemma II.3]{Diggavi_Cover_worst_additive_noise_IT01}.

Since $\mathbf{Q}_0=\frac{P_k}{d_k}\Id_{d_k}$, we can let $\Psi_k=\frac{P_k}{d_k}\mathbf{Q}_Z^{-1}=\left[\begin{array}{cc} \frac{P_k}{d_k\sigma^2} \Id_{d'_k} & \mathbf{0}\\ \mathbf{0} & \frac{P_k}{d_k(\sigma^2 +\sum_{i\neq k} P_i)}  \Id_{N_k-d'_k} \end{array} \right]$, and obtain
\begin{eqnarray}
 \log\det\left( \Id_{N_k} + \bar{\mathbf{H}}_0\mathbf{Q}_0\bar{\mathbf{H}}_0\her \mathbf{Q}_Z^{-1}\right) & = &  \log\det\left( \Id_{N_k} + \bar{\mathbf{H}}_0 \bar{\mathbf{H}}_0\her \Psi_k\right) \\
&=& \log\det\left( \Id_{d_k} + \bar{\mathbf{H}}_0\her \Psi_k\bar{\mathbf{H}}_0 \right) \label{final_log_H0}.
\end{eqnarray}
Substitution of \eqref{final_log_H0} in \eqref{approxgeq} yields 
\begin{equation}
\expect{\mathbf{H}}{I(\mathbf{s}_k;\mathbf{y}_k|\mathbf{H})} \apprge \expect{\bar{\mathbf{H}}_0}{\log\det\left( \Id_{d_k} + \bar{\mathbf{H}}_0\her \Psi_k\bar{\mathbf{H}}_0 \right)}. \label{bound_expectation}
\end{equation}
Thanks to the diagonal structure of $\Psi_k$, the expectation over $\bar{\mathbf{H}}_0$ can be evaluated using results from \cite{Chiani_Ergodic_MI_intf_networks_IT10}, since the right-hand side of \eqref{bound_expectation} is by definition equal to $\mathcal{C}_{\mathrm{SU}} \left( N_k,d_k,\Psi_k \right)$, giving the result in \eqref{eq_thm_MI_bound}.
\end{proof}

\subsection{Projection Receiver}

\label{section_ergodic_MI_projectionreceiver}
We now focus on the projection receiver defined in \eqref{equ_yibar_zf}, and extend the ergodic rate formula of \cite{Guillaud_Tresch_NewcomWS_BCN_09} to the case of receive diversity $d'_i\neq d_i$.

The ergodic rate achieved by user $k$ with the projection receiver under the assumption of Gaussian i.i.d. channel coefficients was characterized in \cite{Guillaud_Tresch_NewcomWS_BCN_09}, based on the observation that since $\mathbf{U}_{k}$ and $\mathbf{V}_{k}$ are truncated unitary matrices, $\bar{\mathbf{H}}_{kk}=\mathbf{U}_{k}\her\mathbf{H}_{kk}\mathbf{V}_{k}$ has Gaussian i.i.d. coefficients of the same variance as $\mathbf{H}_{kk}$, and the noise term $\mathbf{U}_{k}\her\mathbf{n}_k$ is Gaussian i.i.d. with variance $\sigma^2$. Therefore, \eqref{equ_yibar_zf} describes the transmission over a Rayleigh-fading channel of dimensions $d'_k\times d_k$, with ergodic mutual information
\begin{align}
\expect{\mathbf{H}}{I(\mathbf{s}_k;\bar{\mathbf{y}}_k|\mathbf{H})}=
\expect{\bar{\mathbf{H}}_{kk}}{\log\det\left(\Id_{d_k}+\bar{\mathbf{H}}_{kk}\mathbf{Q}_{k}\bar{\mathbf{H}}_{kk}\her\right)}.
\end{align}
Under the assumption of a spatially white transmit signal inside the subspace defined by $\mathbf{V}_{k}$ ($\mathbf{Q}_{k}=\frac{P_k}{d_k }\Id_{d_k}$), we define the SNR $\rho_k=\frac{P_k}{d_k \sigma^2}$ and have directly from \cite{ShinLeeit03} that
\begin{equation} \label{eq_rateexpectation_projreceiver}
\expect{\bar{\mathbf{H}}_{kk}}{ \log\det\left(\Id_{d_k}+\frac{P_k}{d_k \sigma^2} \bar{\mathbf{H}}_{kk}\bar{\mathbf{H}}_{kk}\her\right) } = C(d_k,d'_k,\rho_k),
\end{equation}
where $C(d,d',\rho)$
{\small 
\begin{eqnarray} \label{MIMO_allwhite_ergodic_capa}
&=& e^{\frac{1}{\rho}} \log_2(e)  \sum_{k=0}^{d-1} \sum_{l=0}^{k} \sum_{m=0}^{2l} \Bigg\{ \frac{(-1)^m (2l)! (d'-d+m)!}{2^{2k-m} l! m! (d'-d+l)!}   \left(\begin{array}{c} 2k-2l \\ k-l \end{array} \right)
 \left(\begin{array}{c} 2l + 2d'-2d \\ 2l-m \end{array} \right)
 \sum_{p=1}^{d'-d+m+1} E_p\left(\frac{1}{\rho}\right)
 \Bigg\}
\end{eqnarray} }
and where $E_p(\cdot)$ denotes the exponential integral function of order $p$, i.e.,
\begin{align}
E_p(z)=\int_1^\infty e^{-zx} x^{-p} \textrm{d}x, \quad \textrm{Re}\{z\}>0.
\end{align}

\begin{figure}[ht]
\centering \includegraphics[width=10cm]{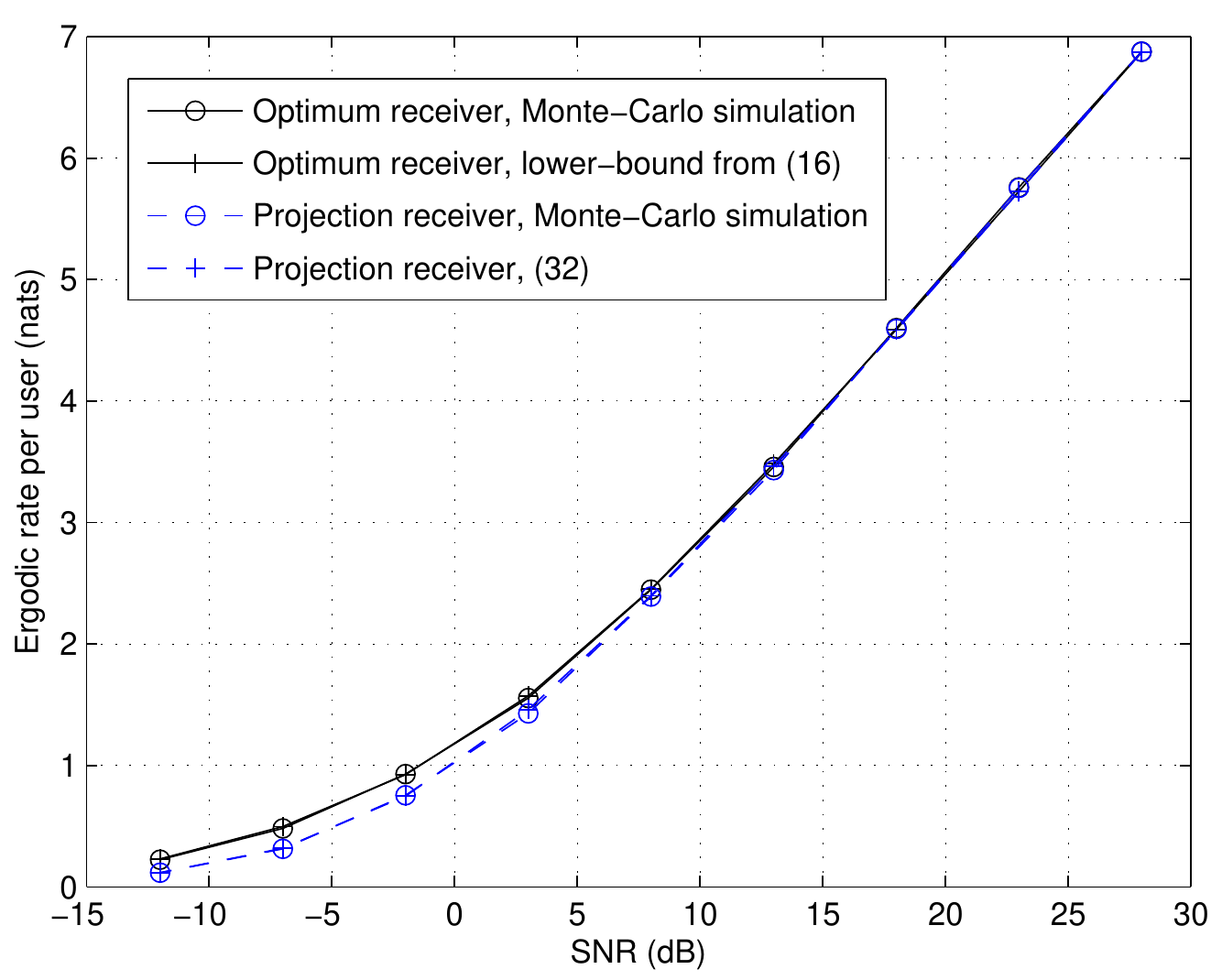}
\caption{Ergodic rate achieved under IA by the optimum and projection receivers, $K=7$ users with $5\times 7$ channels, $d=1, d'=2$.} \label{fig_mc_vs_analytical}
\end{figure}

\subsection{Simulations}
\label{section_simuls_validation}
Fig.~\ref{fig_mc_vs_analytical} depicts the simulated ergodic achievable rates for both the optimum and the projection-based receivers, for a symmetric system with $K=7$ users, each using $\NT=7$ transmit and $\NR=5$ receive antennas, for a range of SNR (the SNR=$P_k/\sigma^2$ is assumed identical for all users). IA is employed with parameters $d=1, d'=2$, i.e. each user receives a single stream with diversity 2.
The results are compared to the analytical formulas established in Sections~\ref{section_ergodic_MI_optimumreceiver} and \ref{section_ergodic_MI_projectionreceiver}. Specifically, the two dashed lines depict the left-hand side (LHS, obtained by Monte-Carlo simulation) and right-hand side (RHS) of \eqref{eq_rateexpectation_projreceiver}, showing an excellent agreement. The solid lines, depicting the LHS and RHS of \eqref{eq_thm_MI_bound}, are essentially identical, showing that the bound of Theorem~\ref{thm_bound_IA} is tight even for relatively low number of interferers ($K-1=6$ here).

\section{Performance Comparison of IA to Distributed, Game-theoretic Approaches}
\label{section_applications}

As an application of the result of Section~\ref{section_ergodic_MI_for_IA}, we now present a comparison of the ergodic rates achievable over the Gaussian MIMO-IC for various DMT points, as well as a comparison of IA with a distributed covariance optimization based on game theory.

\subsection{Ergodic Rates under IA for various DMT}

In this section, we evaluate the influence of the DMT on the achievable ergodic rates established in Sections~\ref{section_sysmodel} and \ref{section_non_square_eqchannel}, both for the optimum receiver and the projection receiver. The results are presented in Fig.~\ref{fig_influence_Dp} for a symmetric system where all the links have $\NT=7$ transmit antennas and $\NR=5$ receive antennas. Depending on $d$ and $d'$, a different number of users can be accommodated. Here, we consider the case of $K=11$ users with $d=d'=1$, as well as the case $K=7$ with $d=1$ and $d'=2$. In both cases, the criterion \eqref{sym_proper_condition} is verified with equality.

 In all cases considered, the advantage of the optimum receiver over the projection receiver is moderate at low SNR, and inexistent at high SNR.
\begin{figure*}[ht]
\centering 
\subfigure[Per user rate]{\includegraphics[width=10cm]{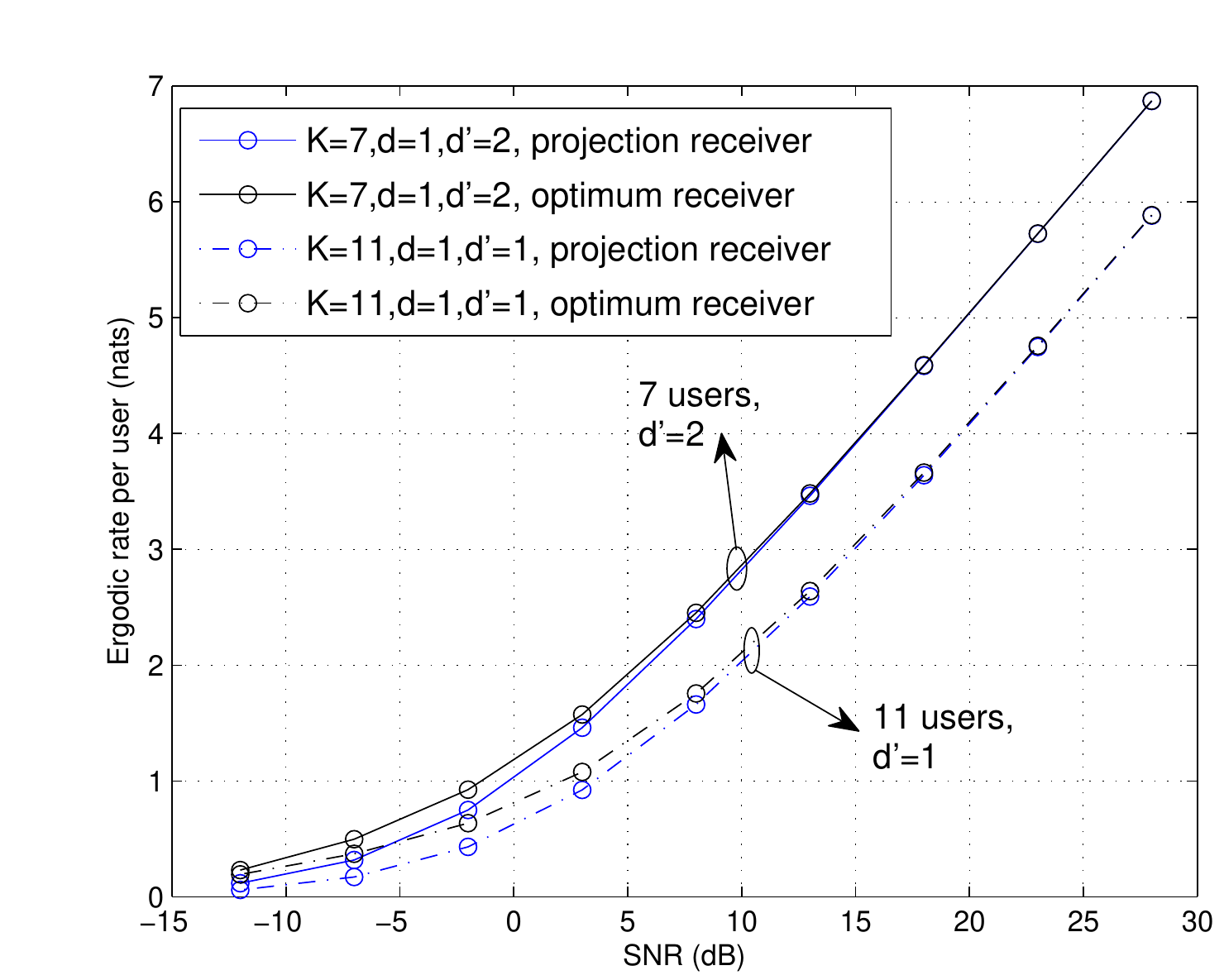} \label{fig_influence_Dp_1}}
\subfigure[Sum-rate]{\includegraphics[width=10cm]{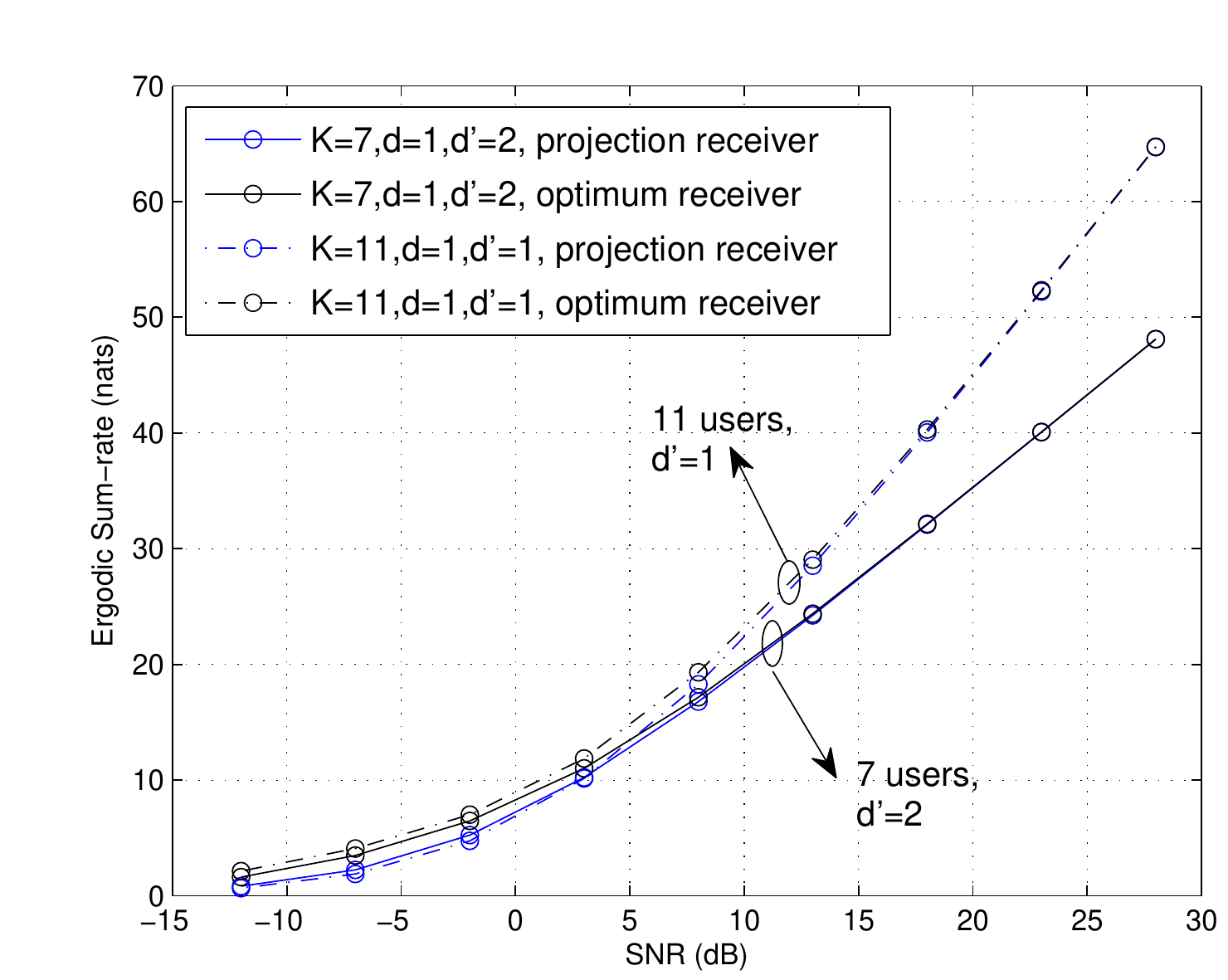} \label{fig_influence_Dp_2}}
\caption{Ergodic rate achievable under IA for symmetric systems with variable number of users and DMT. $\NT=7, \NR=5$. SNR=$P_k/\sigma^2 \ \forall k$.} \label{fig_influence_Dp}
\end{figure*}
In Fig.~\ref{fig_influence_Dp_1}, it is noticeable that the receive diversity is beneficial to the ergodic rate achievable per-user, since the curve for $d'=2$ clearly dominates the $d'=1$ curve (the transmit signal dimension is $d=1$ in both cases).
However, when considering the sum-rate (Fig. \ref{fig_influence_Dp_2}), we come to an opposite conclusion, since the loss in per-user ergodic rate when going from receive diversity $d'=2$ to $d'=1$ is more than offset by the fact that more users ($K=11$) can be supported while still fulfilling the IA conditions when $d'=1$.
Interestingly, in the SNR range of 0 to 5 dB, Fig.~\ref{fig_influence_Dp_1} indicates that it is possible to trade the number of users $K$ for the per-user rate (and presumably the outage performance, although this is not shown here) through the introduction of diversity in the interference-free subspace introduced in Section~\ref{section_non_square_eqchannel}, while the sum-rates (Fig.~\ref{fig_influence_Dp_2}) remain comparable.\\

\subsection{Comparison with distributed, game-theoretic covariance optimization}

In order to illustrate the use of the results of Section~\ref{section_ergodic_MI_for_IA}, we compare the performance achievable by IA to the distributed covariance optimization scheme of Scutari et al. \cite{Scutari_MIMO_Iterative_WTF_IT09}, which is based on a game-theoretic analysis of the choice of the transmit covariances. In the method of \cite{Scutari_MIMO_Iterative_WTF_IT09}, each transmitter does waterfilling based on the interference covariance at his intended receiver resulting from the previous iterations. This yields an inherently distributed algorithm,  which exhibits faster convergence properties than the leakage-based IA algorithm of \cite{Gomadamit08}.
However, it does not always converge to a Nash equilibrium (depending on $\mathbf{H}$). In case of non-convergence, we draw another channel realization and ignore the channel realization in the mutual information average.

The results are presented in Fig.~\ref{fig_ia_vs_gt}. IA is applied for the same two points of the DMT as in the previous section ($K=7$ users with receive diversity $d'=2$, and 11 users with diversity 1 each). 
For SNR above 5dB, IA presents an advantage over the distributed waterfilling game method, both in terms of rate per user and sum-rate. 
\begin{figure*}[ht]
\centering
\subfigure[Per user rate]{\includegraphics[height=8.4cm]{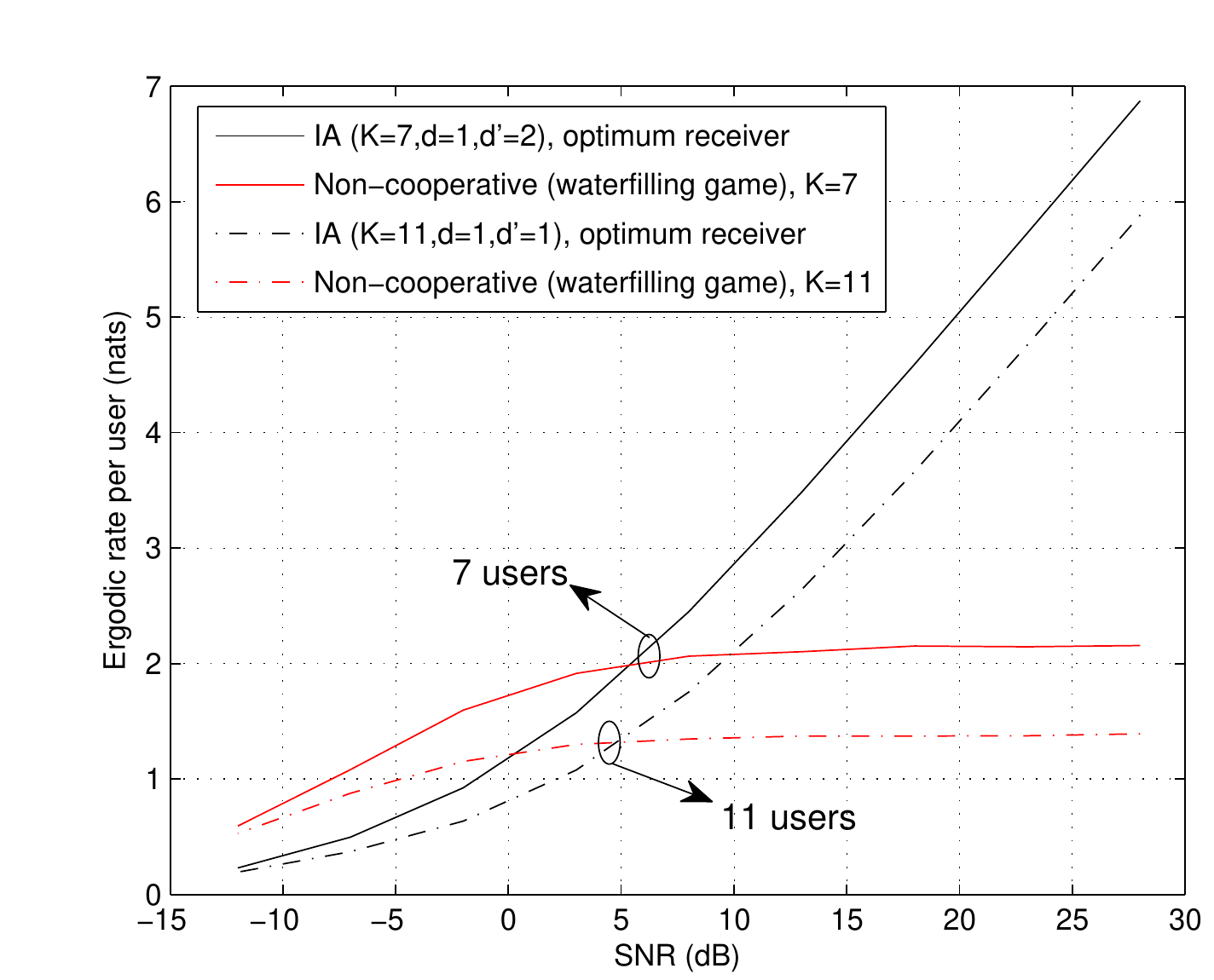} \label{fig_ia_vs_gt_user}}
\subfigure[Sum-rate]{\includegraphics[height=8.3cm]{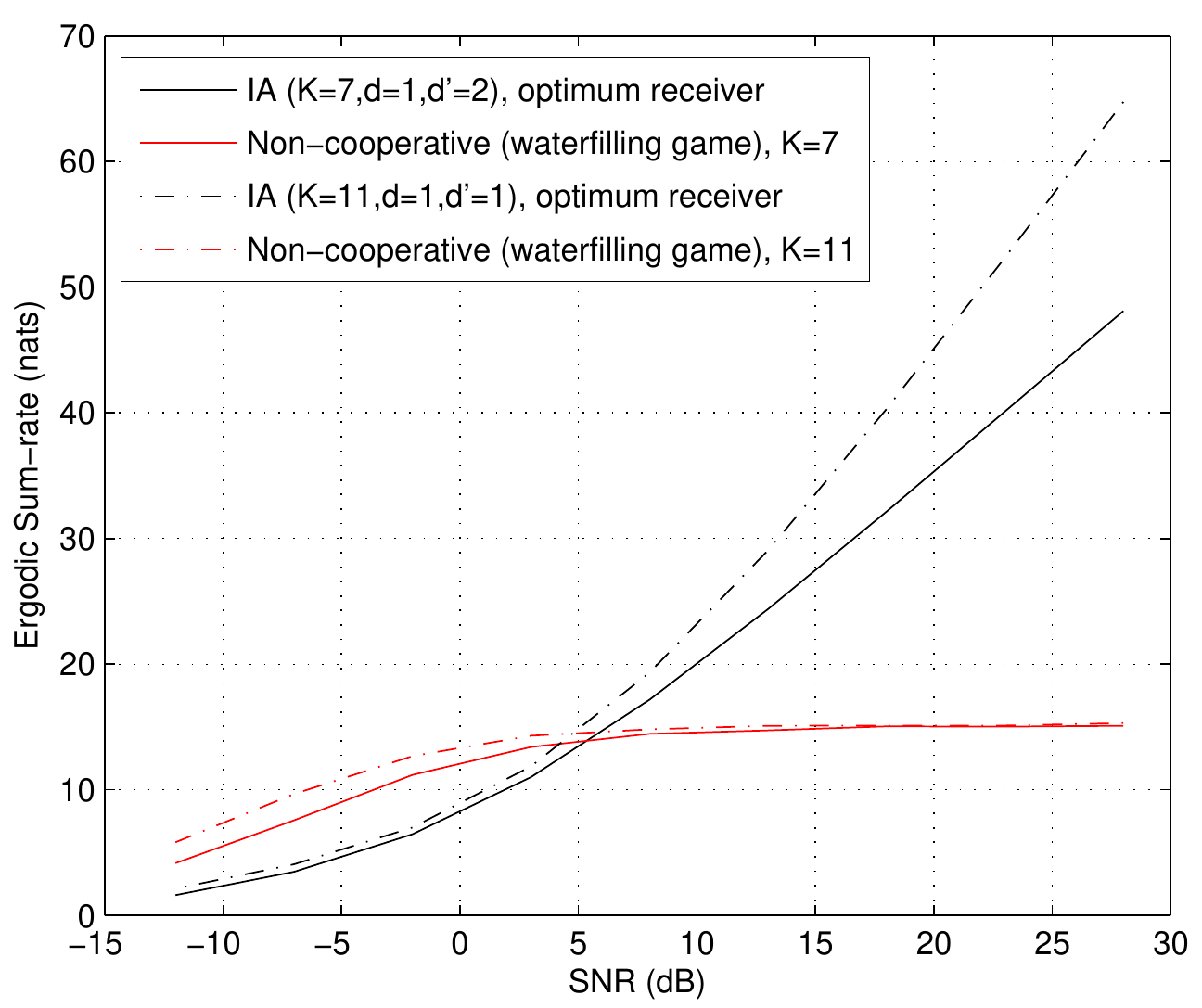} \label{fig_ia_vs_gt_SR}}
\caption{Ergodic mutual information achieved by IA and waterfilling game, symmetric system, $\NT=7, \NR=5$. SNR=$P_k/\sigma^2 \ \forall k$.} \label{fig_ia_vs_gt}
\end{figure*}
Fig.~\ref{fig_ia_vs_gt_user} shows that both approaches benefit from having fewer users in terms of per-user rates. However, in terms of sum-rate (Fig.~\ref{fig_ia_vs_gt_SR}), it is noticeable that the influence of the number of users on the distributed waterfilling game is small, while the IA method benefits more dramatically from having more users, since in that case the high-SNR slope follows the total DoF.  \\

\section{Conclusion}

We introduced a generalization of IA whereby receive diversity can be obtained and exploited inside the interference-free subspace at the receiver, and generalized the achievability criterion of IA to this case, thus defining a multi-user DMT for the interference channel.
In a second step, we obtained closed-form expressions  for the ergodic mutual information achievable by IA over the Gaussian MIMO-IC under the Gaussian i.i.d. fading assumption, for an arbitrary SNR. 
Finally, we compared the performance of IA (at various operating points allowed by our DMT) to the per-user and sum rates achieved by a distributed, game theoretic method, and provided insight on their relative performance according to the considered SNR range.

\section*{Acknowledgment}

This work was sponsored by the Austria Science Fund (FWF) through grant NFN SISE (S10606), and inspired by activities in the European Commission FP7 Newcom++ network of excellence.
Part of this work was performed while the author was with FTW.
FTW's research is supported in part by the COMET competence center program of the Austrian government and the City of Vienna.



%

\bibliographystyle{IEEEtran}
\bibliography{biblio}

\end{document}